\documentclass[aps,physrev,groupedaddress,preprint]{revtex4-2}
\usepackage{amsmath}
\usepackage{amssymb}
\usepackage{physics}
\usepackage{hyperref}
\usepackage{graphicx}
\usepackage{amsthm}
\usepackage{bbold}

\newtheorem{theorem}{Theorem}
\newtheorem{corollary}{Corollary}
\begin{document}


\title{Tighter Error Bounds for the qDRIFT Algorithm}



\date{\today}

\author{I. J. David$^{1,2,4}$}
\email{ian.david@nithecs.ac.za}
\author{I. Sinayskiy$^{1,2}$}
\email{sinayskiy@ukzn.ac.za}
\author{F. Petruccione$^{2,3}$}
\email{petruccione@sun.ac.za}
\affiliation{$^{1}$School of Chemistry and Physics, University of KwaZulu-Natal, Durban 4001, South Africa.}
\affiliation{$^{2}$National Institute for Theoretical and Computational Sciences (NITheCS), Stellenbosch, South Africa.}
\affiliation{$^{3}$School of Data Science and Computational Thinking, Stellenbosch 7604, South Africa.}
\affiliation{$^{4}$Fraqtal Technologies, Durban 4051, South Africa}


\date{\today}

\begin{abstract}
Randomized algorithms such as qDRIFT provide an efficient framework for quantum simulation by sampling terms from a decomposition of the system’s generator. However, existing error bounds for qDRIFT scale quadratically with the norm of the generator, limiting their efficiency for large-scale closed or open quantum system simulation. In this work, we refine the qDRIFT error bound by incorporating Jensen’s inequality and a careful treatment of the integral form of the error. This yields an improved scaling that significantly reduces the number of steps required to reach a fixed simulation accuracy. Our result applies to both closed and open quantum systems, and we explicitly recover the improved bound in the Hamiltonian case. To demonstrate the practical impact of this refinement, we apply it to three settings: quantum chemistry simulations, dissipative transverse field Ising models, and Hamiltonian encoding of classical data for quantum machine learning. In each case, our bound leads to a substantial reduction in gate counts, highlighting its broad utility in enhancing randomized simulation techniques.
\end{abstract}


\maketitle

\section{Introduction}
\label{sec1}
Quantum simulation is widely regarded as one of the most promising applications of quantum computing, first proposed by Feynman and Manin \cite{feynman1982simulating, manin1980computable}. It offers the potential to study complex quantum systems beyond the reach of classical methods \cite{lloyd1996universal}. Among the various techniques developed for quantum simulation \cite{berry2007efficient, berry2015simulating, berry2015hamiltonian, childs2021theory, low2017optimal, low2019hamiltonian, childs2012hamiltonian, childs2019nearly}, randomized methods \cite{childs2019faster,campbell2019random,chen2024randomized, david2024faster,ding2025lower} have emerged as particularly promising due to their simplicity and favourable error scaling properties.

One of the most notable contributions in this direction is the qDRIFT algorithm \cite{campbell2019random} which simulates unitary dynamics generated by sparse Hamiltonians using randomized gate sequences. In contrast to traditional Trotter-based approaches \cite{berry2007efficient,childs2021theory}, which make use of the Trotter-Suzuki product formulas \cite{suzuki1990fractal,suzuki1991general}, qDRIFT avoids the use of Trotter-Suzuki product formulas and makes use of a stochastic process to construct a gate set and achieves an error that is independent of the number of terms in the Hamiltonian. More recently the qDRIFT framework has been extended to Makovian open quantum systems described by Lindblad master equations \cite{david2024faster,chen2024randomized,ding2025lower}, allowing stochastic simulation of dissipative dynamics through random sampling of terms from the generator.

Despite these advances, a significant challenge remains in the scaling of the number of required steps (and hence, the number of quantum gates) needed to achieve a fixed simulation accuracy. Finding tighter error bounds for existing quantum simulation algorithms solves this problem and improves accuracy and efficiency of the algorithm. In both the Hamiltonian and Lindbladian settings, the current bounds scale quadratically with respect to the sum of the coefficients in the generator (or Hamiltonian in the closed system case), leading to large gate counts in practical applications. This scaling becomes especially pronounced when simulating complex many-body systems.

In this work, we present a new proof for the qDRIFT algorithm that tightens its theoretical error bound, thereby reducing the number of steps required to simulate both closed and open quantum systems. Our result improves the dependence of the number of steps needed by qDRIFT with respect to the norms of the terms in the generator. In the closed system case we improve the scaling of the number steps from quadratic in the sum of the coefficients to linear in this parameter. In the open system case we remove the dependence on the maximum of the diamond norms of the terms in the generator. This improvement leads to a refined bound on the number of steps needed by the qDRIFT algorithm. This refinement is achieved by making use of the integral error representation of Taylors theorem as in \cite{childs2021theory} and Jensens inequality \cite{jensen1906fonctions,durrett2019probability}.

To demonstrate the practical impact of our result, we apply the new bound in three representative scenarios: the simulation of molecular electronic structure Hamiltonians relevant to quantum chemistry\cite{wecker2014gate} and compare our new bound to the original qDRIFT bound by estimating the number of steps needed for simulating the molecules studied in \cite{campbell2019random}, dissipative spin dynamics in the transverse field Ising model \cite{hashizume2022dynamical}, and Hamiltonian encoding of high-dimensional data in quantum machine learning \cite{schuld2021machine}. In each case, we show that the new bound leads to substantial reductions in the number of steps, often by several orders of magnitude.

The remainder of this paper is organized as follows. In Section \ref{sec2}, we review the qDRIFT method and its application to both closed and open system dynamics. Section \ref{sec3} presents our improved error bound, including a detailed derivation and corollary for the Hamiltonian case. In Section \ref{sec4}, we demonstrate the advantages of our approach across three application domains. Finally, in Section \ref{sec5}, we conclude with a discussion of the implications of our results and potential future directions.
\section{Background and Motivation}
\label{sec2}

The state of a $d$-dimensional quantum system at time $t$ is described by a $d \times d$ density matrix $\rho(t)$. The density matrix evolves under a quantum channel $\Lambda(t)$, which is a completely positive and trace-preserving (CPTP) map \cite{breuer2002theory,rivas2012open},
\begin{align}
    \rho(t) = \Lambda(t)\rho(0).
\end{align}
This channel satisfies the master equation
\begin{align}
\label{eq:master_eqn}
    \dv{t} \Lambda(t) = \mathcal{L} \Lambda(t),
\end{align}
where $\mathcal{L}$ is the Liouvillian, also referred to as the generator. When the quantum channel describes unitary evolution of a closed system, the Liouvillian takes the form
\begin{align}
    \mathcal{L} \rho = -i[H, \rho],
\end{align}
where $H$ is the system Hamiltonian. When the channel describes a Markovian open quantum system, the Liouvillian takes the well-known Gorini-Kossakowski-Sudarshan-Lindblad (GKSL) form \cite{gorini1976completely,lindblad1976generators},
\begin{align}
    \mathcal{L} \rho = -i[H, \rho] + \sum_k \gamma_k \left( L_k \rho L_k^\dagger - \frac{1}{2} \{L_k^\dagger L_k, \rho\} \right),
\end{align}
where $H$ is the Hamiltonian, $\gamma_k \geq 0$ are decay rates, and $\{L_k\}$ are jump operators. Solving the master equation \eqref{eq:master_eqn} yields
\begin{align}
    \Lambda(t) = \exp(t \mathcal{L}).
\end{align}

For the remainder of this work, we express the generator as a linear combination of the form
\begin{align}
\label{eq:generator}
    \mathcal{L} = \sum_{k=1}^{n} \lambda_k \mathcal{L}_k,
\end{align}
where $\lambda_k \geq 0$ and $n$ is the number of terms. This unified form enables a consistent treatment of generators for both closed and open system dynamics. We proceed as follows:

\begin{itemize}
    \item[(i)] \textbf{Closed systems:} For unitary evolution generated by $-i[H, \rho]$, we express the Hamiltonian as
    \begin{align}
        H = \sum_s h_s H_s,
    \end{align}
    where $H_s$ are Hermitian and normalized, and the coefficients $h_s$ can be chosen to be non-negative. The generator then takes the form \eqref{eq:generator} by setting $\lambda_k = h_k$ and $\mathcal{L}_k(\rho) = -i[H_k, \rho]$.

    \item[(ii)] \textbf{Open systems:} For the GKSL generator, we apply the same decomposition to the Hamiltonian part. For the dissipative part, we set $\lambda_k = \gamma_k$ and define $\mathcal{L}_k(\rho) = L_k \rho L_k^\dagger - \frac{1}{2} \{L_k^\dagger L_k, \rho\}$.
\end{itemize}

In quantum simulation, the goal is to efficiently simulate the channel $\Lambda(t)$ on a quantum computer. Various methods exist for simulating such evolutions \cite{berry2007efficient,berry2015simulating,childs2012hamiltonian,low2017optimal}, but we focus on the qDRIFT method, which applies to both closed \cite{campbell2019random} and open systems \cite{chen2024randomized,david2024faster}. This method relies on expressing the generator $\mathcal{L}$ as a sum of simple terms that are easy to exponentiate.

qDRIFT simulates the evolution $\Lambda(t)$ up to a specified precision $\epsilon > 0$ using a stochastic process based on the so-called qDRIFT channel \cite{campbell2019random}. The total simulation time $t$ is divided into $r$ time steps of length $\omega = t\Gamma / r$, where
\begin{align}
    \Gamma = \sum_{k=1}^{n} \lambda_k
\end{align}
is the sum of the coefficients in \eqref{eq:generator}. The qDRIFT channel is a convex mixture of the exponentials of the individual generator terms. The probability of applying $\exp(\tau \mathcal{L}_k)$ is given by
\begin{align}
    p_k = \frac{\lambda_k}{\Gamma},
\end{align}
so that the qDRIFT channel is defined as
\begin{align}
\label{eq:qdrift_channel}
    \mathcal{E}(\tau) = \sum_{k=1}^{n} p_k \exp(\tau \mathcal{L}_k).
\end{align}
This channel represents a random process where each exponential is applied with probability $p_k$, resulting in an average evolution that stochastically “drifts” towards $\Lambda(t)$. Since each operation is sampled independently, the process is fully Markovian, and one may analyze it in terms of repeated applications of a single-step random channel.

The error in approximating $\Lambda(t)$ by $\mathcal{E}(\omega)^r$ is given by
\begin{align}
    \|\Lambda(t) - \mathcal{E}(\omega)^r\|_{\diamond} \leq \frac{(t \Gamma \Omega)^2}{r} \exp\left( \frac{t \Gamma \Omega}{r} \right),
\end{align}
where $\|\cdot\|_{\diamond}$ denotes the diamond norm \cite{watrous2004notes,watrous2009semidefinite}, and $\Omega = \max_k \|\mathcal{L}_k\|_{\diamond}$. In practice, this inequality is used to estimate the number of steps $r$ required to achieve a desired accuracy $\epsilon$ by bounding the right-hand side above by $\epsilon$ and solving for $r$. Approximating the exponential term by 1, as done in \cite{campbell2019random}, leads to the simplified bound
\begin{align}
    r \geq \left\lceil \frac{(t \Gamma \Omega)^2}{\epsilon} \right\rceil.
\end{align}

A key feature of the qDRIFT method, as reflected in this bound, is that the required number of steps $r$ is independent of the number of generator terms $n$. This makes qDRIFT particularly well suited for systems with many generator terms, such as two dimensional spin lattices with dissipation \cite{rota2018dynamical} or electronic structure Hamiltonians \cite{wecker2014gate}. However, the quadratic dependence on $\Gamma$ and $\Omega$ may still lead to large $r$ in practice.

In this work, we re-derive the qDRIFT error bound for both closed and open quantum systems to analyse its scaling with $\Gamma$ and $\Omega$ more precisely. Our refined analysis, based on the integral form of the Taylor expansion error and Jensen's inequality \cite{jensen1906fonctions,durrett2019probability}, yields significantly tighter bounds. For closed systems, the new bound scales linearly with $\Gamma$ and is independent of $\Omega$. For open systems, the bound avoids dependence on the maximum diamond norm of the generator terms. Furthermore, our formulation eliminates exponential pre-factors, enabling a direct bound on $r$ without requiring additional approximations. Section \ref{sec3} presents and proves our main result, demonstrating how qDRIFT can be accelerated.

\section{Improved Error Bounds for the qDRIFT Protocol}
\label{sec3}
In this section, we present a modified qDRIFT simulation protocol that achieves a tighter convergence bound compared to the standard approach. We formalize this in Theorem~\ref{theorem1} and provide a detailed proof. The result enables a more efficient scaling of number of steps $r$ with respect to dependence on the norm of generator.

\begin{theorem}
\label{theorem1}
    Given the generator $\mathcal{L}$, as in equation (\ref{eq:generator}), of a quantum channel $\Lambda(t)$ as well as the simulation time $t\geq 0$, a positive integer $r$. Then,
    \begin{align}
        \|\Lambda(t)- \mathcal{E}(t/r)^{r\Gamma}\|_{\diamond}\leq \frac{t^{2}\sum_{k}\lambda_{k}\|\mathcal{L}_{k}\|_{\diamond}^{2}}{r},
    \end{align}
where $\mathcal{E}$ is defined in equation (\ref{eq:qdrift_channel}) and $r$ is bounded by,
\begin{align}
    r \geq  \left\lceil \frac{t^{2}}{\epsilon}\sum_{k=1}^{n}\lambda_{k}\|\mathcal{L}_{k}\|_{\diamond}^{2}\right\rceil.
\end{align}
\end{theorem}

\begin{proof}
    We start by considering the channel $\Lambda(t)=\exp(t\mathcal{L})$ we see that $\Lambda(t)=\Lambda(t/r\Gamma)^{r\Gamma}$. This allows us to write,
    \begin{align}
        \|\Lambda(t)-\mathcal{E}(t/r)^{r\Gamma}\|_{\diamond}&=\|\Lambda(t/r\Gamma)^{r\Gamma}-\mathcal{E}(t/r)^{r\Gamma}\|_{\diamond}\ ,\\
        \label{eq:lemma_1_step}
        &\leq r\Gamma \|\Lambda(t/r\Gamma)-\mathcal{E}(t/r)\|_{\diamond}
    \end{align}
where in (\ref{eq:lemma_1_step}) we have used Lemma 1 from \cite{david2024faster}. Now we need to bound the norm $\|\Lambda(t/r\Gamma)-\mathcal{E}(t/r)\|_{\diamond}$, we shall do this using the integral representation of the error for the Taylor expansion \cite{apostol1991calculus, bartle1992introduction, childs2021theory}. We start by defining $\eta=t/r\Gamma$ expanding $\Lambda(\eta)$ to first order with a second order error term,
\begin{align}
    \Lambda(\eta)= \mathbb{1}+\eta\mathcal{L}+\int_{0}^{\eta} dx \left(\eta-x\right) \ \dv[2]{\eta}\left(e^{\eta \mathcal{L}}\right)_{\eta=x}.
\end{align}
Next we expand the exponential in the qDRIFT channel to first order with second order error here we define $\tau=t/r$ then,
\begin{align}
    \mathcal{E}(\tau)=\sum_{k=1}^{n}p_{k}\left( \mathbb{1}+\tau\mathcal{L}_{k}+\int_{0}^{\tau}dy \ (\tau-y)\dv[2]{\tau}\left(e^{\tau\mathcal{L}_{k}}\right)_{\tau=y}\right).
\end{align}
Now we can look at the difference $\Lambda(\eta)-\mathcal{E}(\tau)$ where we see that the terms up to first order cancel since $\sum_{k}p_{k}=1$ and,
\begin{align}
\sum_{k=1}^{n}p_{k}\tau\mathcal{L}_{k}=\sum_{k=1}^{n}\frac{\lambda_{k}}{\Gamma}\frac{t}{r}\mathcal{L}_{k}=\frac{t}{r\Gamma}\sum_{k=1}^{n}\lambda_{k}\mathcal{L}_{k}=\eta \mathcal{L}.
\end{align}
What remains is,
\begin{align}
\label{eq:error_terms_1}
   & \Lambda(\eta)-\mathcal{E}(\tau)=
    \int_{0}^{\eta} dx \left(\eta-x\right) \ \dv[2]{\eta}\left(e^{\eta \mathcal{L}}\right)_{\eta=x}\nonumber\\
    &- \sum_{k=1}^{n}p_{k} \int_{0}^{\tau}dy \ (\tau-y)\dv[2]{\tau}\left(e^{\tau\mathcal{L}_{k}}\right)_{\tau=y}.
\end{align}
The derivatives in (\ref{eq:error_terms_1}) can we evaluated and we have that $\dv[2]{\eta}\left(e^{\eta \mathcal{L}}\right)\big|_{\eta=x}=\mathcal{L}^{2}e^{x\mathcal{L}}$ and $\dv[2]{\tau}\left(e^{\tau\mathcal{L}_{k}}\right)\big|_{\tau=y}=\mathcal{L}_{k}^{2}e^{y\mathcal{L}_{k}}$, which can be substituted into (\ref{eq:error_terms_1}),
\begin{align}
\label{eq:error_terms_2}
    \Lambda(\eta)-\mathcal{E}(\tau)&=\int_{0}^{\eta} dx \left(\eta-x\right)\mathcal{L}^{2}e^{x\mathcal{L}} \nonumber\\
    &- \sum_{k=1}^{n}p_{k} \int_{0}^{\tau}dy \ (\tau-y)\mathcal{L}_{k}^{2}e^{y\mathcal{L}_{k}}.
\end{align}
Let us analyse each term in (\ref{eq:error_terms_2}). In the first term we make the substitution $x=u\eta$ which leads to,
\begin{align}
    \int_{0}^{\eta} dx \left(\eta-x\right)\mathcal{L}^{2}e^{x\mathcal{L}}&=\int_{0}^{1} du \ \eta^{2}(1-u)\mathcal{L}^{2}e^{u\eta\mathcal{L}},
\end{align}
Next we replace $\eta=t/r\Gamma$ and rewrite the integrand in a more convenient form so that,
\begin{align}
\label{eq:simp_int_1}
    \int_{0}^{1} du \ & \eta^{2}(1-u)\mathcal{L}^{2}e^{u\eta\mathcal{L}}\nonumber\\
    &=\int_{0}^{1}du \ \left(\frac{t}{r}\right)^{2}(1-u)\left(\frac{\mathcal{L}}{\Gamma}\right)^{2}e^{ut\mathcal{L}/r\Gamma}.
\end{align}
Analysing the second term in (\ref{eq:error_terms_2}), we make the substitution $y=u\tau$ and replace $\tau=t/r$ so that,
\begin{align}
    \sum_{k=1}^{n}p_{k} \int_{0}^{\tau}dy & \ (\tau-y)\mathcal{L}_{k}^{2}e^{y\mathcal{L}_{k}}\nonumber\\
    &=\sum_{k=1}^{n}p_{k}\int_{0}^{1}du \ \left(\frac{t}{r}\right)^{2}(1-u)\mathcal{L}_{k}^{2}e^{ut\mathcal{L}_{k}/r}.
\end{align}
Then we use the properties of the integral and sums so that we can write,
\begin{align}
\label{eq:simp_int_2}
    \sum_{k=1}^{n}p_{k}\int_{0}^{1}du &\ \left(\frac{t}{r}\right)^{2}(1-u)\mathcal{L}_{k}^{2}e^{ut\mathcal{L}_{k}/r}\nonumber\\
    &=\int_{0}^{1}du \ \left(\frac{t}{r}\right)^{2}(1-u) \sum_{k=1}^{n}p_{k}\mathcal{L}_{k}^{2}e^{ut\mathcal{L}_{k}/r}
\end{align}

Now taking the diamond norm on both sides of (\ref{eq:error_terms_1}) and using the simplified integrals in (\ref{eq:simp_int_1}) and (\ref{eq:simp_int_2}) as well as the sub-additive, sub-multiplicative and integral property of the diamond norm we have,
\begin{align}
    &\|\Lambda(t/r\Gamma)-\mathcal{E}(t/r)\|_{\diamond}\leq \nonumber\\
    &\int_{0}^{1}du \ \left(\frac{t}{r}\right)^{2}(1-u)\left\|\frac{\mathcal{L}}{\Gamma}\right\|_{\diamond}^{2}\|e^{ut\mathcal{L}/r\Gamma}\|_{\diamond}\nonumber\\
    &+\int_{0}^{1} du \ \left(\frac{t}{r}\right)^{2}(1-u)\sum_{k=1}^{n}p_{k}\|\mathcal{L}_{k}\|_{\diamond}^{2}\|e^{ut\mathcal{L}_{k}/r}\|_{\diamond}
\end{align}
Now using the fact that $u\in [0,1]$ as well as $t/r\Gamma \geq 0$ and $t/r \geq 0$ we have that both exponentials are quantum channels and hence,
\begin{align}
    \|e^{ut\mathcal{L}/r\Gamma}\|_{\diamond}\leq 1, && \|e^{ut\mathcal{L}_{k}/r}\|_{\diamond}\leq 1.
\end{align}
Now we need to bound $\|\mathcal{L}/\Gamma\|_{\diamond}^{2}$.  We observe that,
\begin{align}
    \frac{\mathcal{L}}{\Gamma}=\sum_{k=1}^{n}\frac{\lambda_{k}}{\Gamma}\mathcal{L}_{k}=\sum_{k=1}^{n}p_{k}\mathcal{L}_{k}.
\end{align}
By the sub-additive property of the diamond norm we know that,
\begin{align}
    \left\|\frac{\mathcal{L}}{\Gamma}\right\|_{\diamond}\leq \sum_{k=1}^{n}p_{k}\|\mathcal{L}_{k}\|_{\diamond}
\end{align}
Taking the sqaure of both sides we have,
\begin{align}
  \left\|\frac{\mathcal{L}}{\Gamma}\right\|_{\diamond}^{2}\leq \left(\sum_{k=1}^{n}p_{k}\|\mathcal{L}_{k}\|_{\diamond}\right)^{2}
\end{align}
Now we need to bound this quantity, we do this by using Jensen's inequality \cite{jensen1906fonctions,durrett2019probability}, which states that for a real valued convex function $f$ defined on an interval $A \subset \mathbb{R}$ with $z_{1},z_{2},...,z_{k} \in A$ and $w_{1},w_{2},...,w_{k}\geq 0$, then we have,
\begin{align}
    f\left(\frac{\sum_{i=1}^{k}w_{i}z_{i}}{\sum_{j=1}^{k}z_{j}}\right)\leq\frac{\sum_{i=1}^{k}w_{k}f(z_{k})}{\sum_{j=1}^{k}z_{j}}.
\end{align}
Now we observe that the function $x\mapsto x^{2}$ is convex and that the input of a convex sum $\sum_{k=1}^{n}p_{k}\|\mathcal{L}_{k}\|_{\diamond}$ and Jensens inequality leads to,
\begin{align}
    \left\|\frac{\mathcal{L}}{\Gamma}\right\|_{\diamond}^{2}\leq \left(\sum_{k=1}^{n}p_{k}\|\mathcal{L}_{k}\|_{\diamond}\right)^{2} \leq \sum_{k=1}^{n}p_{k}\|\mathcal{L}_{k}\|_{\diamond}^{2}.
\end{align}
Putting everything together and using the fact that $\int_{0}^{1}du \ (1-u)=1/2$ we have,
\begin{align}
\label{eq:small_step_bound}
    \|\Lambda(t/r\Gamma)-\mathcal{E}(t/r)\|_{\diamond}\leq \frac{t^{2}}{r^{2}}\sum_{k=1}^{n}p_{k}\|\mathcal{L}_{k}\|_{\diamond}.
\end{align}
Using (\ref{eq:small_step_bound}) with (\ref{eq:lemma_1_step}) we have,
\begin{align}
\label{eq:main_bound}
    \|\Lambda(t)-\mathcal{E}(t/r)^{r\Gamma}\|_{\diamond}&\leq \frac{t^{2}\Gamma}{r}\sum_{k=1}^{n}p_{k}\|\mathcal{L}_{k}\|_{\diamond}^{2}\nonumber\\
    &=\frac{t^{2}}{r}\sum_{k=1}^{n}\lambda_{k}\|\mathcal{L}_{k}\|_{\diamond}^{2}.
\end{align}
Bounding (\ref{eq:main_bound}) by $\epsilon$, solving for $r$ and taking the ceiling function we achieve the desired bound for $r$ and thus completing the proof.
\end{proof}

We now apply Theorem~\ref{theorem1} to the special case of unitary dynamics generated by a Hamiltonian. This is the setting originally considered in the qDRIFT proposal, where each $\mathcal{L}_k$ generates the evolution of a closed system. The following corollary shows how the newly derived bound on $r$ from Theorem \ref{theorem1}, improves upon the bound in \cite{campbell2019random} by yielding a linear scaling in the $\Gamma$, that is, the sum of the coefficients in the Hamiltonian.

\begin{corollary}
\label{corollary1}
    For a generator $\mathcal{L}$ that describes closed system evolution with the Hamiltonian $H=\sum_{s}h_{s}H_{s}$ such that $h_{s}\geq 0$ and $H_{s}$ is normalised and Hermitian the qDRIFT method requires,
    \begin{align}
        r \geq \left\lceil \frac{4t^{2}\Gamma}{\epsilon} \right\rceil
    \end{align}
    where $\epsilon >0 $ is the precision. 
\end{corollary}

\begin{proof}
    Using the result of Theorem \ref{theorem1}. we see that we will need to compute $\sum_{k=1}^{n}\lambda_{k}\|\mathcal{L}_{k}\|_{\diamond}^{2}$. We observe that for closed system evolution $\mathcal{L}_{k}(\rho)=-i[H_{k},\rho]=-i(H_{k}\rho - \rho H_{k})$, we compute the diamond norm to find,
    \begin{align}
        \|\mathcal{L}_{k}\|_{\diamond}\leq 2\|H_{k}\|\leq 2,
    \end{align}
    since we chose $H_{k}$ so that they are normalised. Using this in the error in Theorem \ref{theorem1}. we have,
    \begin{align}
        \epsilon \geq \frac{4t^{2}\sum_{k=1}^{n}\lambda_{k}}{r}=\frac{4t^{2}\Gamma}{r}.
    \end{align}
    Solving for $r$ and considering the fact that $r$ is a positive integer we take the ceiling function leading to the bound,
    \begin{align}
        r \geq \left\lceil \frac{4t^{2}\Gamma}{\epsilon} \right\rceil.
    \end{align}
\end{proof}

The improved error bound shown in Theorem~\ref{theorem1} highlights the advantage of the modified qDRIFT protocol. By carefully analysing second-order contributions, we demonstrate that simulation precision can be achieved with fewer samples, thereby reducing the overall simulation cost while maintaining accuracy. In section \ref{sec4} we demonstrate this by comparing the result in corollary \ref{corollary1} to what was obtained in the original qDRIFT paper \cite{campbell2019random}.

\section{Applications}
\label{sec4}
\begin{figure*}[h]
    \centering
    \includegraphics[width=\textwidth]{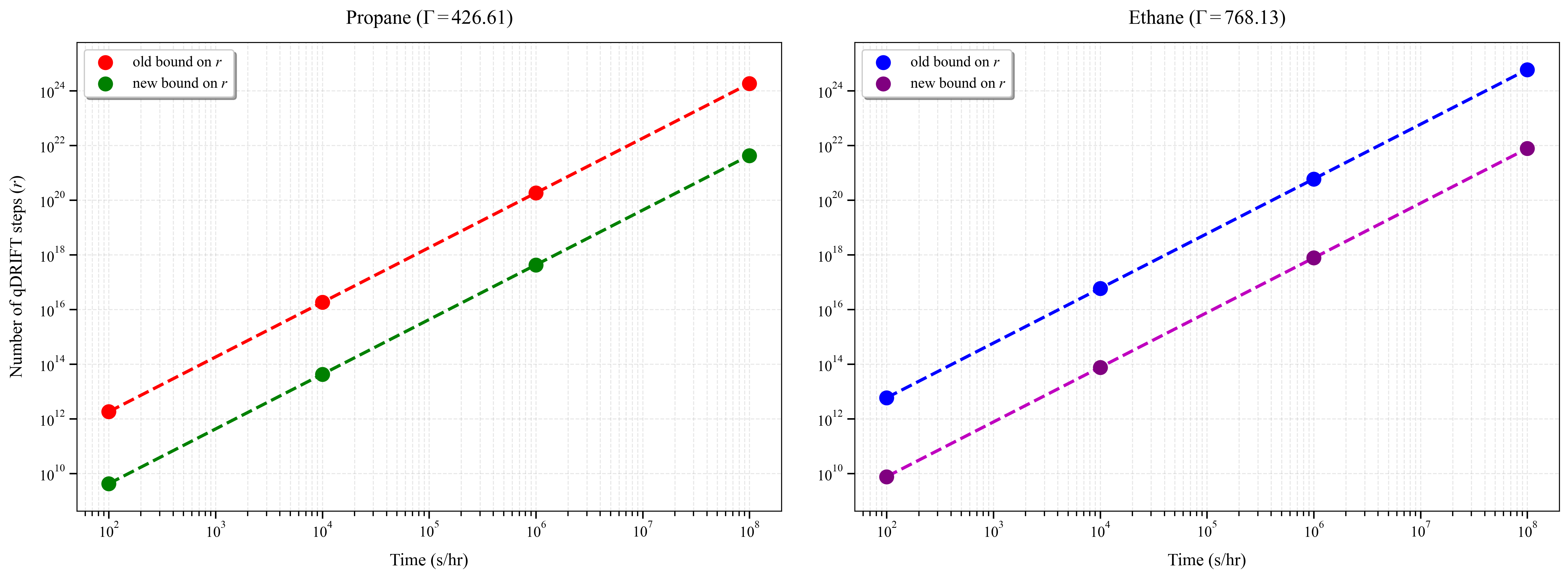}
    \caption{Comparison of the number of steps $r$ needed to simulate propane (left) and ethane (right) to a fixed precision $\epsilon = 10^{-3}$ over a range of simulation times, given in units s/hr which quantifies the rate of change in seconds over each hour, using the original qDRIFT bound and our improved bound. The improvement grows with simulation time, leading to reductions of several orders of magnitude.}
    \label{fig:chemistry}
\end{figure*}

In this section, we illustrate how our improved bound can be applied to a range of problems where the qDRIFT algorithm is relevant. First, we demonstrate how the new bound reduces the number of gates required to simulate electronic structure Hamiltonians, improving upon the results in the original qDRIFT paper \cite{campbell2019random}. Next, we show how it leads to more efficient simulation of a dissipative transverse field Ising model \cite{hashizume2022dynamical} as we need fewer steps $r$. Finally, we apply our bound to Hamiltonian-based data encoding, a technique of increasing importance in quantum machine learning \cite{schuld2021machine}.

\subsection{Demonstration via chemistry simulations}

Here we demonstrate the practical impact of our improved bound on $r$ in the context of simulating electronic structure Hamiltonians using qDRIFT. Specifically, we compare the number of steps required to simulate the molecules propane and ethane over a range of simulation times, using both the original qDRIFT bound from \cite{campbell2019random} and the improved bound derived in Corollary \ref{corollary1}. These examples are representative of the types of quantum chemistry problems considered in the original qDRIFT proposal.

Figure~\ref{fig:chemistry} shows $r$ plotted against $\text{Time}$ (in seconds/Hour), where the $x$ and $y$ axes are in $\log_{10}$ scale, comparing the number of steps required by the original bound (dashed red and blue lines) to those required by our improved bound (dashed green and purple lines). As expected, our new bound results in a consistent and significant reduction in $r$ across all time scales. This improvement scales proportionally with the sum of the coefficients $\Gamma$, leading to reductions of up to 2 orders of magnitude. 

To illustrate, at a simulation time of $10^6$ seconds, the number of steps needed to simulate propane and ethan drops from approximately $10^{20}$ to less than $10^{18}$. Such dramatic reductions in the number of steps correspond directly to fewer quantum gates, making these simulations more tractable.

This example underscores the practical utility of our theoretical result. By reducing the number of required steps without compromising precision, our bound enables more efficient quantum simulations of realistic chemical systems—highlighting the relevance of our analysis to quantum chemistry and beyond.

\subsection{Transverse Field Ising Model With Dissipation}

\begin{figure}[h!]
    \centering
    \includegraphics[width=\linewidth]{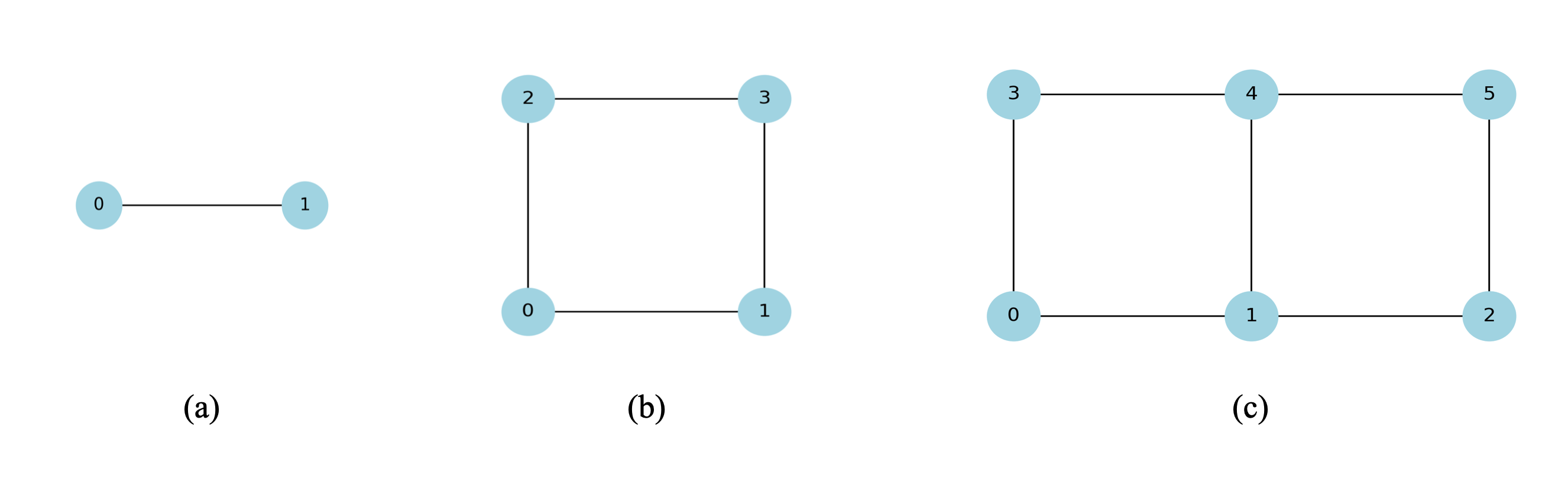}
    \caption{(a)-(c) shows the lattice configuration for the transverse field Ising model for two, four and six spins respectively. Here the nodes are spins and the edges represent the neighbour-neighbour interacting spins.}
    \label{fig:ising-lattices}
\end{figure}

We now demonstrate the effectiveness of our newly derived qDRIFT bound in the simulation of open quantum spin systems, using the transverse field Ising model with local dephasing as a representative example. This model, widely studied in both closed and open quantum dynamics \cite{hashizume2022dynamical}, consists of spin-1/2 particles arranged on a lattice, interacting via nearest-neighbour coupling and subject to transverse magnetic fields and local dephasing noise.

We consider lattice configurations with two, four, and six spins, as illustrated in Figure~\ref{fig:ising-lattices}. The Hamiltonian for this model is given by
\begin{align}
    H = -J\sum_{\langle i,j\rangle}Z_{i}Z_{j} - h\sum_{j}X_{j},
\end{align}
where $J$ denotes the strength of the nearest-neighbour coupling, $h$ is the transverse field amplitude in the $x$-direction, and $\langle i,j\rangle$ indicates summation over unique neighbouring spin pairs. Throughout this work, we use the values $J = 1$ and $h = 0.5$.

The dissipation is introduced through a set of local jump operators of the form
\begin{align}
    L_{j} = \sqrt{\gamma} \ Z_{j},
\end{align}
where $\gamma$ denotes the dephasing rate, and we set $\gamma = 0.1$. The full GKSL generator governing the system dynamics is then expressed as
\begin{align}
    \mathcal{L}(\rho) = iJ\sum_{\langle i,j \rangle}[Z_{i}Z_{j},\rho] + ih\sum_{j}[X_{j},\rho] + \gamma\sum_{j}(Z_{j}\rho Z_{j} - \rho).
\end{align}
The number of terms $n$ in this Liouvillian depends on the specific lattice connectivity: for two spins, $n = 5$; for four spins, $n = 12$; and for six spins, $n = 19$.

To assess the practical implications of our improved qDRIFT bound, we apply the stochastic simulation strategy to the open-system evolution generated by $\mathcal{L}$. For each system size, we compute the number of qDRIFT steps $r$ required to achieve a fixed simulation accuracy $\epsilon = 10^{-6}$ over a time $t = 4$. We compare the performance of the original and new bounds, using norm calculations based on the component operator weights.

\begin{figure}[h!]
    \centering
    \includegraphics[width=\linewidth]{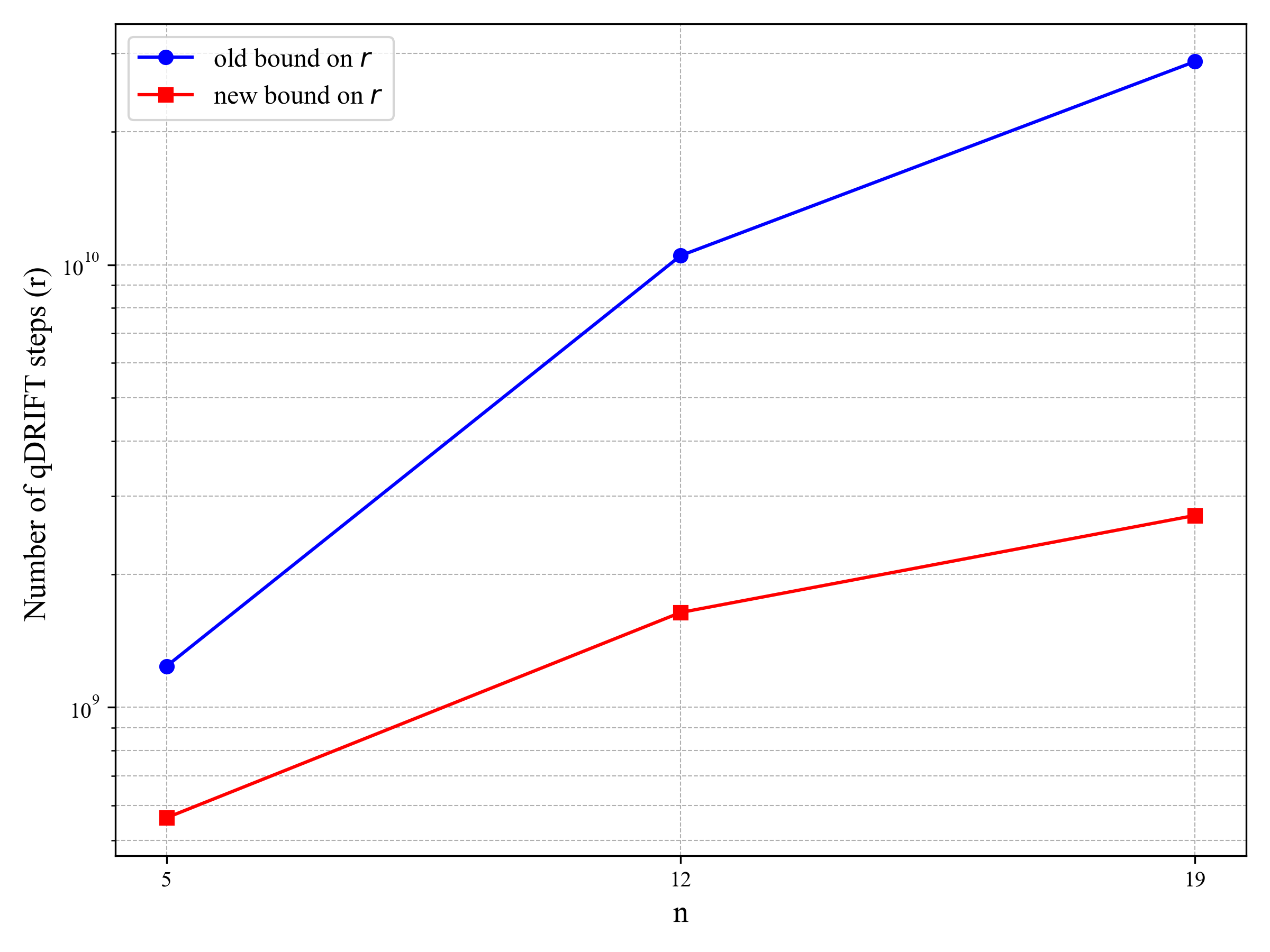}
    \caption{Comparison of the number of qDRIFT steps $r$ required to simulate the open transverse field Ising model using the old and new qDRIFT bounds. The simulation was performed for systems with 2, 4, and 6 spins. The new bound (red) yields a significant reduction in the required number of steps compared to the original bound (blue).}
    \label{fig:ising-qdrift-results}
\end{figure}

The results, presented in Figure~\ref{fig:ising-qdrift-results}, demonstrate that the new bound leads to a marked reduction in the number of required qDRIFT steps across all system sizes. While the number of required steps under the old bound scales steeply with the number of spins. The number of steps $r$ is $2.87 \times 10^{10}$ for six spins—our bound reduces this to $2.71\times 10^{9}$ steps. This improvement corresponds to 1 order of magnitude reduction in simulation cost, and underscores the practical advantage of the new bound for simulating dissipative quantum systems.

Such improvements are crucial when resources are limited and efficient implementations of open-system dynamics are essential. Our results indicate that the new qDRIFT bound not only improves theoretical scaling but also translates into tangible runtime savings in realistic open-system simulations.

\subsection{Hamiltonian Encoding in Quantum Machine Learning}

In this section, we demonstrate that our newly derived qDRIFT bound also yields significant improvements when applied to Hamiltonian encoding of large input data vectors, which is a key subroutine in certain quantum machine learning protocols \cite{schuld2021machine}.

Consider a classical data vector $\vec{x} \in \mathbb{R}^{N}$. To encode this vector into a quantum system, we construct a Hamiltonian of the form
\begin{align}
    H = \sum_{i=1}^{N}|x_{i}|\tilde{P}_{i},
\end{align}
where $x_i$ is the $i$-th component of the vector $\vec{x}$ and $\tilde{P}_{i} = \mathrm{sign}(x_i) P_i$. Here, each $P_i$ is a Hermitian and unitary operator constructed as a tensor product of single-qubit Pauli operators ($I$, $X$, $Y$, $Z$). These Pauli strings $P_i$ serve as basis operators in the space of $2^l \times 2^l$ Hermitian matrices, with $l = \lceil \log_2 N \rceil$ being the number of qubits required to represent an $N$-dimensional vector. This embedding ensures that the constructed Hamiltonian is both efficiently simulated and faithful to the structure of the input data. Writing the Hamiltonian this way also ensures that the requirements for qDRIFT are satisfied, that is, the coefficients $|x_{i}|$ are positive and $\|\tilde{P}_{i}\| = 1$.

For large-scale machine learning problems where $N$ is large (e.g., $N \sim 10^2$ or more), simulating the time evolution under $H$ is computationally intensive. The qDRIFT algorithm offers a randomized, structure-exploiting simulation method that samples terms in $H$ proportional to their norm, and approximates $e^{-iHt}$ with a product of randomly selected unitaries. The number of qDRIFT steps, $r$, required to simulate $H$ to within a fixed error $\epsilon$ depends on a norm-related bound. Our new bound modifies this scaling significantly.

To illustrate this, we computed the number of required qDRIFT steps using both the original and our improved bound as a function of the input dimension $N$. We used randomly generated data vectors $\vec{x}$ and calculated $r$ for each $N$ using both bounds, fixing the error tolerance $\epsilon = 10^{-7}$ and total simulation time $t=1$. The results are shown in Figure~\ref{fig:hamiltonian-encoding-results}.

\begin{figure}[h!]
    \centering
    \includegraphics[width=\linewidth]{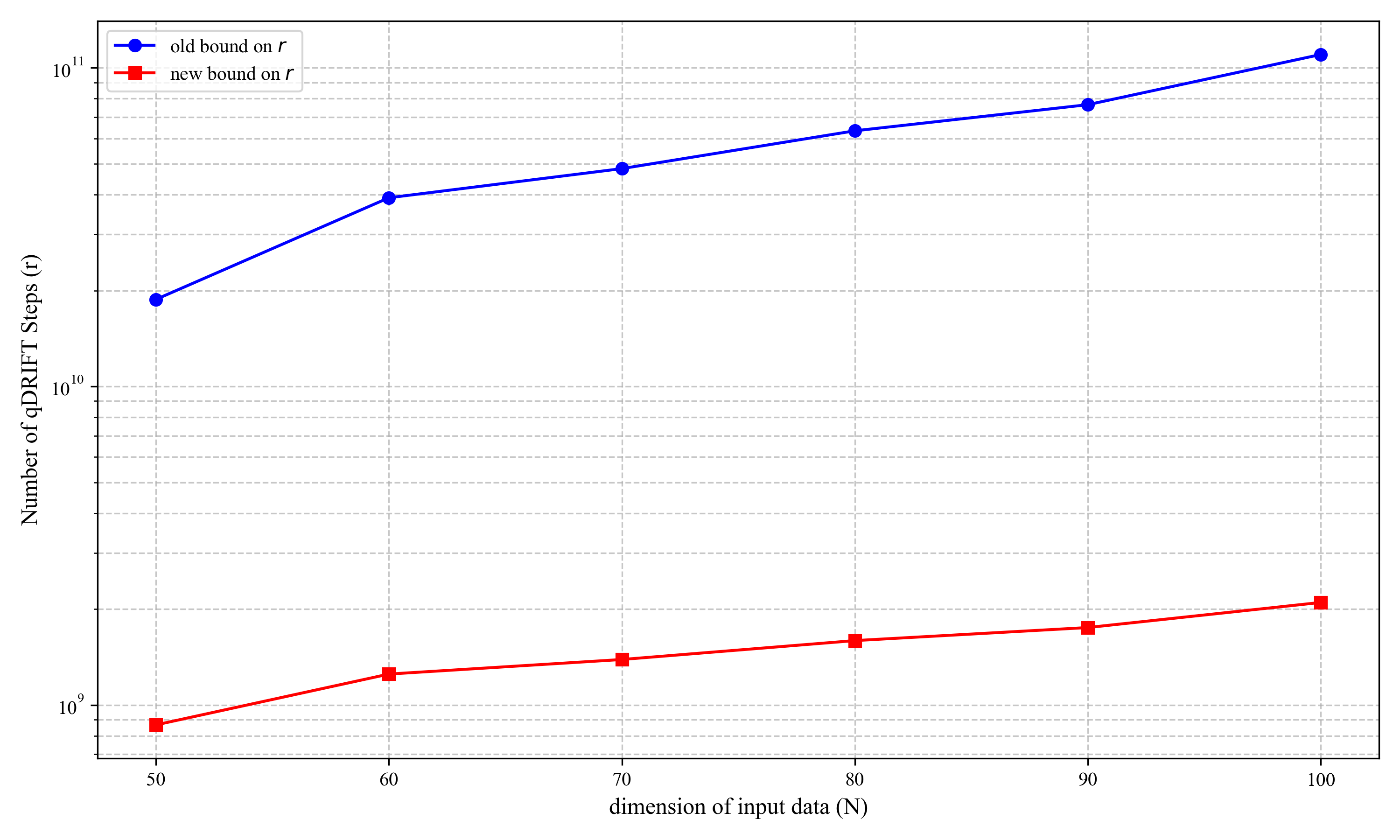}
    \caption{Comparison of the number of qDRIFT steps $r$ required for Hamiltonian encoding of classical data vectors, as a function of vector dimension $N$. The blue curve shows the result using the old qDRIFT bound, while the red curve uses our improved bound. The improvement is more than two orders of magnitude for all $N$ tested.}
    \label{fig:hamiltonian-encoding-results}
\end{figure}

From the plot, we observe that the old qDRIFT bound yields step counts on the order of $10^{11}$ for $N=100$, whereas our new bound consistently remains around $10^9$ in this regime. This is a reduction of roughly two orders of magnitude in simulation cost.

\section{Conclusion}
\label{sec5}
Efficient quantum simulation remains a central challenge in quantum algorithm design, particularly for both unitary and non-unitary dynamics. While randomized methods such as qDRIFT offer a promising framework for compiling simulations with reduced gate counts \cite{campbell2019random,childs2019faster}, their error bounds—and hence resource requirements—can scale unfavourably in many practical scenarios where the sum of the coefficients $\Gamma$ is large. This motivates a re-examination of the theoretical foundations of such algorithms to improve their asymptotic performance and broaden their applicability.
In this work, we have introduced a new theoretical analysis of the qDRIFT algorithm, deriving a tighter error bound that leads to a substantial reduction in the number of samples (and hence gates) required for simulation. Our new proof generalizes naturally to both Hamiltonian evolution and Lindbladian dynamics, thereby addressing both closed and open quantum systems within a unified framework. As a direct corollary, our result strengthens the original Hamiltonian-only qDRIFT analysis, offering improvements in both theoretical scaling and practical performance by scaling linearly in the sum of the coefficients $\Gamma$.
To demonstrate the practical benefits of our bound, we applied it to three representative settings. First, in the simulation of molecular systems such as propane and ethane, our bound reduces the gate count by several orders of magnitude when compared to the original qDRIFT scaling \cite{campbell2019random}. Second, we considered a dissipative transverse-field Ising model on lattices of increasing size, showing consistent improvements in the simulation cost under our bound. Lastly, we explored Hamiltonian encoding of classical data vectors for quantum machine learning, where our result allows for efficient encoding of high-dimensional vectors with significantly fewer quantum gates, preserving simulation accuracy.
Together, these demonstrations highlight the broad utility of our improved bound, showing that it not only advances the theoretical analysis of randomized simulation algorithms but also yields tangible resource savings across diverse quantum applications. The reduction in circuit depth and sample complexity can directly impact the feasibility of simulating complex systems on future quantum hardware and inform algorithm design in data-intensive quantum tasks.
Future work could explore trying to achieve similar scaling for other quantum simulation algorithms for open quantum systems, we would also like to explore the use of this bound in thermal state preparation as in \cite{chen2024randomized}.

\section*{Acknowledgements}
The authors would like to thank Ms S. M. Pillay for proof reading the manuscript and her insightful discussions. This work is based upon research supported by the
National Research Foundation of the Republic of South
Africa. IJD is the Chief Operations Officer and Co-founder of fraqtal technologies Pty (Ltd). FP is the Chair of the Scientific Board and Co-Founder of
QUNOVA computing. The authors declare no other competing interests.

\bibliography{references}

\end{document}